\documentclass[journal,10pt]{IEEEtran}

\usepackage{cite}
\usepackage{color} 
\usepackage{stfloats}

\ifCLASSINFOpdf
   \usepackage[pdftex]{graphicx}

\else
\fi
\usepackage{cite}
\usepackage{amsmath,amsthm,amssymb,graphicx,algorithm,algorithmic}
\usepackage{makecell}
\usepackage{subfigure}
\usepackage{float}
\usepackage{hyperref}
\newtheorem{theorem}{\bf Theorem}

\usepackage{color}   
\usepackage{array}
\usepackage{multicol}
\begin{document}

\title{Refracting Reconfigurable Intelligent Surface Assisted URLLC for Millimeter Wave High-Speed Train Communication Coverage Enhancement}

\author{

Changzhu~Liu,~\IEEEmembership{Graduate~Student~Member,~IEEE},
Ruisi~He,~\IEEEmembership{Senior~Member,~IEEE}, \\
Yong Niu,~\IEEEmembership{Senior Member,~IEEE},
Shiwen~Mao,~\IEEEmembership{Fellow,~IEEE},\\
Bo~Ai,~\IEEEmembership{Fellow,~IEEE},
and Ruifeng Chen

\thanks{
Copyright (c) 20xx IEEE. Personal use of this material is permitted. However, permission to use this material for any other purposes must be obtained from the IEEE by sending a request to pubs-permissions@ieee.org.
This work is supported by the Fundamental Research Funds for the Central Universities under Grant 2022JBQY004, the National Natural Science Foundation of China under Grant 62431003, 62271037, 62221001, and 62231009, the State Key Laboratory of Advanced Rail Autonomous Operation under Grant RCS2022ZZ004, and the Beijing Natural Science Foundation under Grant L232042. \emph{(Corresponding author:Ruisi He; Yong Niu.)}

C. Liu, R. He, and B. Ai are with the School of Electronics and Information Engineering, and the Frontiers Science Center for Smart High-speed Railway System, Beijing Jiaotong University, Beijing 100044, China (e-mails: changzhu{\_}liu@bjtu.edu.cn; ruisi.he@bjtu.edu.cn; boai@bjtu.edu.cn).

Y. Niu is with the School of Electronics and Information Engineering, Beijing Jiaotong University, Beijing 100044, China, and also with the National Mobile Communications Research Laboratory, Southeast University, Nanjing 211189, China (e-mail: niuy11@163.com).

S. Mao is with the Department of Electrical and Computer Engineering, Auburn University, Auburn, AL 36849-5201 USA (e-mail: smao@ieee.org).

R. Chen is with the Institute of Computing Technology, China Academy of Railway Sciences Corporation Limited, Beijing 100081, China (e-mail: ruifeng{\_}chen@126.com).

              }
\thanks{
              }
              }

\markboth{IEEE Transactions on Vehicular Technology}
{}

\maketitle

\begin{abstract} %
High-speed train (HST) has garnered significant attention from both academia and industry due to the rapid development of railways worldwide. Millimeter wave (mmWave) communication, known for its large bandwidth is an effective way to address performance bottlenecks in cellular network based HST wireless communication systems. However, mmWave signals suffer from significant path loss when traversing carriage, posing substantial challenges to cellular networks. To address this issue, reconfigurable intelligent surfaces (RIS) have gained considerable interest for its ability to enhance cell coverage by reflecting signals toward receiver. Ensuring communication reliability, a core performance indicators of ultra-reliable and low-latency communications (URLLC) in fifth-generation systems, is crucial for providing steady and reliable data transmissions along railways, particularly for delivering safety and control messages and monitoring HST signaling information. In this paper, we investigate a refracting RIS-assisted multi-user multiple-input single-output URLLC system in mmWave HST communications. We propose a  sum rate maximization problem, subject to base station beamforming constraint, as well as refracting RIS discrete phase shifts and reliability constraints. To solve this optimization problem, we design a joint optimization algorithm based on alternating optimization method. This involves decoupling the original optimization problem into active beamforming design and packet error probability optimization subproblem, and discrete phase shift design subproblems. These subproblems are addressed exploiting Lagrangian dual method and the local search method, respectively. Simulation results demonstrate the fast convergence of the proposed algorithm and highlight the benefits of refracting RIS adoption for sum rate improvement in mmWave HST networks.


\end{abstract}

%
\IEEEpeerreviewmaketitle

\section{Introduction}
\IEEEPARstart{H}{igh-speed} train (HST), characterized by its exceptional advantages of high speed, low energy consumption, high safety, and high comfortness, has developed rapidly worldwide. The total mileage of HST has surpassed a remarkable $40,000$ kilometers in China, ranking the first in the world \cite{r1}. As one of the most popular green energy transportation infrastructure, numerous countries are racing to develop or upgrade their HST systems to transition to green energy due to its convenience and reduced carbon emissions \cite{r2,r1new,r2new}. To support the increasing demands of intelligent HST systems,the Global System for Mobile Communications-Railway (GSM-R) is progressively being replaced by the fifth-generation mobile networks for Railway (5G-R) due to 5G's huge advantages of supporting high mobility scenarios, large bandwidth, and low latency, to more effectively ensure the application demands of intelligent HST than the narrowband GSM technology \cite{r3,r4new,r5new}. Among the pivotal technologies of 5G, the ultra-reliable and low-latency communications (URLLC) technology can provide low end-to-end latency and ultra highly reliable communication for HST services such as security, control messaging, and monitoring \cite{r4}. The key performance indicator of URLLC is that with an end-to-end transmission latency of less than $1$ ms and a packet error probability (PEP) on the order of $10^{-5}$, which is critical for the development of future smart HST communication systems \cite{r5}. In HST wireless communications, train-to-ground data transmission not only needs to provide safety-critical railway signal information, but also support various high data rate services. To meet the growing demand for high data rates and huge bandwidth, the development of higher frequency bands is urgent. Millimeter wave (mmWave) technology is proposed to enhance the train-to-ground communications and presents a promising opportunity for future smart HST communication systems \cite{r6}.

As one of the key technologies of 5G, mmWave frequency bands covers from 30 GHz to 300 GHz and have abundant spectrum resources \cite{t1,t2,t3,t4}. Currently, the $3^{\mathrm{rd}}$ generation partnership project divides the frequency range of the 5G mmWave band into 24.25--71.0 GHz, known as frequency range 2 (FR2) \cite{3GPP}. This technology is expected to offer several Gbps transmission data rates for HST communications, enhancing user experiences with applications such as video conferencing, live broadcasting, and online gaming \cite{r7}. Currently, mmWave and massive multiple-input multiple-output (MIMO) techniques are investigated to provide reliable communications with an over-the-air latency of a few milliseconds and extremely high throughput \cite{r8}. However, the high frequency leads to higher phase noise, necessitating larger subcarrier spacing to alleviate its effects. Larger subcarrier spacing can achieve shorter transmission time intervals through scalable digital parameters, making mmWave communications well-suited for URLLC due to their high operating frequency. Moreover, there are two key aspects that enable mmWave communications to support URLLC effectively. The first is beamforming. Operating in the mmWave frequency bands enables small-sized antenna arrays with huge amounts of elements to be used in the user device, forming extremely narrow beams and  facilitate pencil-beam directional transmissions. Significant directional gain can improve the range of mmWave transmissions and alleviate co-channel interference. The second aspect is multi-connectivity. This technology offers both advantages and disadvantages. One major disadvantage is the need for a central processing unit to control all distributed base stations (BSs). Additionally, it requires cooperation between distributed BSs, which necessitates low-error rate shared links. Implementing multi-connectivity also demands complex coordination and management mechanisms, additional signaling overhead for maintenance, and the deployment of more antennas and radio frequency front-ends on user equipment and BSs. Despite these challenges, the benefits of multi-connectivity technology are significant. The high penetration attenuation at mmWave frequencies poses a major barrier to satisfactory communication performance when encountering obstacles. Multi-connectivity effectively addresses this challenge by enabling users to associate with multiple BSs simultaneously, thereby improving the reliability of the communication link. It can reduce network latency by dispersing traffic and reducing the load on a single BS. Additionally, multi-connectivity provides backup paths, allowing for quick switching if the primary path fails, thus maintaining low-latency communications. To meet the high reliability requirements of URLLC, the system must incorporate a redundant design. Multi-connectivity offers link redundancy, which is essential for such high-reliability needs. Given the short signal propagation range of the mmWave frequency band, multi-connectivity expands the effective coverage range through collaboration among multiple BSs, thereby satisfying the coverage requirements of URLLC. By enabling packets to be transmitted over multiple paths to the target user device, multi-connectivity significantly reduces the packet loss rate due to time and buffer overflow constraints. In summary, by combining these technologies, mmWave URLLC has a tremendous potential to meet the stringent demands for high reliability and low latency.

Recently, the reconfigurable intelligent surface (RIS) has garnered significant interest for its ability to control and re-engineer the wireless propagation environment, thereby ensuring that received signals possess the desired property. This capacity markedly enhances the spectrum efficiency and coverage of wireless communications \cite{r9,r9new,r9new2,r9new3}. The integration of RIS introduces a level of intelligence into the wireless communication channel \cite{r10new,r11new}. A RIS is a metasurface composed of numerous low-cost passive reflecting elements, which are configured and adjusted by a programmable controller to alter the phase and/or amplitude response of the metasurface, thereby changing the reflection behavior of incident waves \cite{r10}. The goal of this operation is to effectively enhance the received signal at a specific receiver location, thereby improving the overall system performance. Additionally, the control overhead for the RIS is significantly reduced by leveraging a fast backhaul link from the BS. Using relatively fast PIN diodes, it takes only a few micro-seconds to configure the phase shift of the RIS elements, meeting the URLLC requirements. In contrast to conventional relay-assisted communication systems, the RIS incurs no additional communication delay, as it avoids analog to digital domain conversion through radio frequency chains and eliminates the necessity for complex signal processing and forwarding processes \cite{r11}. Furthermore, the utilization of RIS can fulfill the rigorous latency and reliability demands, ultimately enhancing the quality of URLLC transmissions. Consequently, RIS technology can play a key role in URLLC wireless communications \cite{r10,r12}, and it can be effectively applied in URLLC short packet transmissions (SPT) in the finite blocklength (FBL) regime to ameliorate the performance of high-speed train (HST) communication systems by enhancing the received signal quality and guaranteeing reliable communication.

In recent years, several studies on RIS-assisted HST communications \cite{r13,r14,r15,r16}. Most of these studies considered the deployment of RIS panels at the fixed location on the railway side. However, such deployment limits the service time for high mobility users \cite{r17}. To solve this problem, Wang {\it{et al.}} \cite{r18} proposed deploying RIS on HST for in-train user tracking, leveraging their reflective characteristics to alleviate delay spread. However, this method did not effectively resolve the complex issue of signal attenuation induced by train carriages. Moreover, most existing works on RIS have focused on reflective RIS, assuming co-located transmitters and receivers on one side of the RIS, and that the incident signal is fully reflected, unable to penetrate the surface \cite{r19,r20}. To enhance the flexibility of RIS deployment, the concept of refractive RIS is introduced. Refractive RIS allows all incident signals to pass through the RIS surface, effectively reconstructing channels for blocked users and significantly improving transmission performance \cite{r17,r21,r22,r23}.

Motivated by the aforementioned facts, this study investigates a refracting RIS-aided URLLC system in the mmWave HST scenario. In this setup, a refracting RIS is deployed on the train windows to refract signals towards the users, ensuring reliable communication from the base station (BS) to multiple reliability-critical users during mobility. The sum rate is used as the quality-of-service (QoS) metric. This study considers downlink multiple-input single-output (MISO) communications, where the BS is equipped with multiple antennas serving multiple single antenna users on board. The major contributions of this study are summarized as follows:
\begin{itemize}
\item This study considers a novel refracting RIS-aided multi-user (MU) MISO downlink URLLC system in the mmWave HST scenario to enhance HST reliability coverage, where a refracting RIS is deployed on the train window to create robust wireless links for HST communications.
\item Based on the channel model, a sum rate maximization problem considering the constraints on BS transmit power, RIS phase shifts and reliability is proposed. The coupled optimization variables present a challenge for finding closed-form solutions. Therefore, this study decomposes it into two subproblems and exploit an iterative alternating optimization scheme to address them.
\item The active beamforming optimization and PEP optimization subproblem is non-concave. This study transforms it into a convex optimization problem and address it using Lagrangian dual method. The refracting RIS discrete phase shift optimization subproblem is addressed with a local search method. 
\item This study compared the performance of the proposed algorithm with using ideal phase shift, Shannon rate with ideal phase shift method, Shannon rate method, binary search, random phase shift method, and without RIS deployment. The results represent that the proposed algorithm can significantly improve the sum rate performance under various system parameters, such as transmit power,  number of RIS elements, number of users, number of antennas, number of RIS quantization bits, blocklength, maximum PEP, speed of HST and Rician K-factor. The results validated that the deployment of refracting RIS can effectively improve sum rate of HST communication systems.
\end{itemize}

The rest of the paper is organized as follows. The related work is reviewed in Section II. Section III describes the system model and  formulates the sum rate maximization problem. In Section IV, the sum rate maximization algorithm is proposed. Section V presents simulation results. The final section are conclusions.

{\textbf{\textit{Notations}}}: Italic letters denote scalars. Bold-face lower-case and upper-case letters denote vectors matrices respectively. For $\mathbf{x} = \left(x_1,\cdots,x_n\right)^T$, $\mathrm{diag}\left(\mathbf{x}\right)$ denotes a diagonal matrix of the size $n\times n$ with $x_1,\cdots,x_n$ on the diagonal.  $\mathbf{H} \in \mathbb{C} ^{x\times y}$ denotes that $\mathbf{H} $ is the space of $x\times y$ complex-valued matrices. Let $\left(\cdot\right)^T$ and $\left(\cdot\right)^H$ denote the transpose and Hermitian transpose operations, respectively. $\sim$ stands for ``distributed as". $ \otimes$ represents the Kronecker product.

\section{Related Work}

\subsection{URLLC for mmWave communications}
URLLC-based mmWave communications have been extensively studied due to their potential for high throughput applications. MmWave communications, with their abundant spectrum resources, have become a key technology for enabling high throughput URLLC. Feng {\it{et al.}} \cite{r24} studied URLLC in the  mmWave band in the smart factory scenario, and proposed a vision-aided URLLC framework that eliminates the overhead associated with channel training and beam sweeping. The proposed method of blockage prediction and reference signal received power prediction can assist BS handover to avoid communication interruption. Adamu {\it{et al.}} \cite{r25} introduced a hybrid approach to evaluate bit-error and ergodic capacity performance. Liu {\it{et al.}} \cite{r26} investigated the coexistence of enhanced mobile broadband (eMBB) and URLLC in mmWave communications, a throughput maximization problem was formulated to satisfy the reliability and latency demands of URLLC users. Moreover, multi-connectivity technology had been shown to significantly impact the reliability and latency demands for URLLC users, as demonstrated by the findings in \cite{r27}. A user association scheme focused on throughput maximization was proposed and solved using queuing theory and the matrix geometric method.URLLC in mmWave massive MIMO communication systems was also studied, with a network utility maximization problem under probabilistic reliability and latency constraints being addressed through the Lyapunov technique \cite{r28}. Dias {\it{et al.}} \cite{r29} showed the benefits of sliding window network coding for URLLC mmWave networks using a testbed. Ibrahim {\it{et al.}} \cite{r30} proposed the use of beamforming repeaters to enhance reliable communication in the mmWave band, solving an overall scheduling latency minimization problem with a low-complexity greedy algorithm. However, these studies did not address URLLC with mmWave in high mobility scenarios, such as vehicular Networks, unmanned aerial vehicles, and HST. This paper investigates URLLC and mmWave communications in the HST scenario, aiming to achieve reliable communications for HST.

\subsection{URLLC for High Mobility Communication}
Supporting URLLC in high mobility scenarios, such as vehicular networks, unmanned aerial vehicles, and HST, is a challenging task that has received increasing attention \cite{r31,r32,r33,r34,r35,r36new,r36,r37}. Samarakoon {\it{et al.}} \cite{r31} proposed a power allocation algorithm based on federated learning and Lyapunov optimization to minimize the system power consumption of vehicular users in vehicle-to-everything networks while satisfying low latency and high reliability requirements. Guo {\it{et al.}} \cite{r32} considered the time-varying nature of the channel gain and proposed a sum ergodic capacity maximization problem while guaranteeing the latency violation probability for vehicle-to-vehicle (V2V) links. They derived the latency violation probability for V2V links. Nayak {\it{et al.}} \cite{r33} studied a link adaptation model based on Markov chain to predict three performance indicators: end-to-end link latency, throughout, and block error rate (BLER), which were evaluated by Monte Carlo simulations in various mobility scenarios. Kurma {\it{et al.}} \cite{r34} investigated  URLLC in full-duplex communication with an adaptive amplify-and-forward/decode-and-forward relaying protocol, and derived the expression of outage probability and BLER. Fang {\it{et al.}} \cite{r35} considered a vehicular-to-infrastructure network with the in-band full-duplex backhauling scheme, and formulated a resource allocation problem by jointly vehicular user equipments association optimization, resource block assignment and power allocation. The application of sparse vector code (SVC) to  support URLLC in high mobility scenarios was studied in \cite{r36new}. The authors in \cite{r36} proposed a sparse superimposed vector transmission scheme, which is different from SVC, and developed an iterative interference cancellation algorithm for channel estimation and an iterative data-aided algorithm for date decoding. Zhang {\it{et al.}} \cite{r37} investigated the coexistence of URLLC and eMBB in HST communications. A total power minimization problem was proposed and solved with a greedy algorithm. However, the aforementioned researches did not consider mmWave technologies in high mobility scenarios. Few studies have focused on the applications of URLLC in HST communication systems. How to serve high mobility URLLC remains an open issue. This study investigated URLLC in mmWave and high mobility scenarios, specially the HST scenario, and aim to improve the reliability of HST communication systems.

\subsection{RIS-assisted URLLC}
The emerging RIS technology is considered to be an effective solution to enhance the throughput and reliability of wireless communications \cite{r38}. There have been many studies on RIS-aided URLLC in recent years. Hashemi {\it{et al.}} \cite{r39} analyzed the average PEP and average data rate in RIS-assisted URLLC systems with instantaneous channel state information (CSI). Additionally, a RIS-assisted URLLC system featuring non-linear energy harvesting with high reliability and low delay constraints was investigated in \cite{r40}, and the expressions for the BLER was derived. Almekhlafi {\it{et al.}} \cite{r41} studied the application of RIS in enhancing both eMBB and URLLC services, and formulated two optimization problems, that is, maximizing the eMBB sum rate by joint the RIS passive beamforming design and a multi-objective resource allocation problem aimed at maximizing URLLC packet allowance while minimizing eMBB loss rate. Futher, Hashemi {\it{et al.}} \cite{r42} studied a RIS-aided short packet communication system, and a multi-objective optimization problem was formulated for maximizing the FBL rate while minimizing the utilized channel blocklengths. Chandra {\it{et al.}} \cite{r43} investigated a RIS-assisted MU URLLC system in the presence of electromagnetic interference, and derived expressions for the capacity, system throughput, outage probability, and BLER. Moreover, Zhang {\it{et al.}} \cite{r44} investigated multiple RIS-assisted URLLC systems, and proposed a sum rate maximization problem. Abughalwa {\it{et al.}} \cite{r45} considered a downlink RIS-assisted URLLC system in the FBL regime, and proposed a geometric mean maximization by joint design of the transmission beamforming and RIS phase shifts. However, these studies mainly focus on reflective RIS and lack research on refracting RIS. This study aims to fill this gap by exploring the application of refracting RIS-assisted URLLC in high-mobility scenarios, especially in HST communication systems, to improve reliability and performance.

In summary, there are currently no works on URLLC for refracting RIS-assisted mmWave HST communications. This study addresses this gap by proposing a sum rate problem for HST communication systems assisted by a refracting RIS. To tackle this optimization problem, this study designs a joint optimization algorithm based on alternating optimization method. This involves decoupling the original problem into two subproblems: active beamforming design and PEP optimization subproblem, and discrete phase shift design subproblem. These subproblems are then addressed using general convex approximations and the local search method, respectively.

\section{System Model}
\subsection{Deployment Scenario}
This study considers a MU-MISO downlink URLLC system between BS and $M$ user on board in the mmWave HST scenario with refracting RIS-assisted in F{}ig.~\ref{fig:1}, where BS is equipped with a uniform planar array (UPA) of $ N_B =  N_{x_{\mathrm{B}}} \times  N_{y_{\mathrm{B}}}$ antennas and single antenna users indexed by $ m\in \mathcal{M}\triangleq  \left\{1,\cdots,M\right\}$. A RIS with a uniform planar array (UPA) of $F =  F_{x_{\mathrm{R}}} \times  F_{y_{\mathrm{R}}}$ refracting elements is mounted on the window. Since there has been studied proposed various effective methods for CSI acquisition of RIS-assisted communication systems and high mobility scenarios \cite{t22,t23,t24,t25}, it is reasonable to suppose that the CSI of all channels is perfectly known at both the BS and RIS. For ease of exposition, this study supposes that HST moves toward the BS from the left of BS with speed $v$, and then gradually moves away from the BS. This study assumes that each time slot length is smaller than the coherence time, the channel can be considered constant during each transmission time slot, thus the CSI stays consistent during the transmission time slot.
\begin{figure}[!t]
  \centering
  {\includegraphics[scale=0.4]{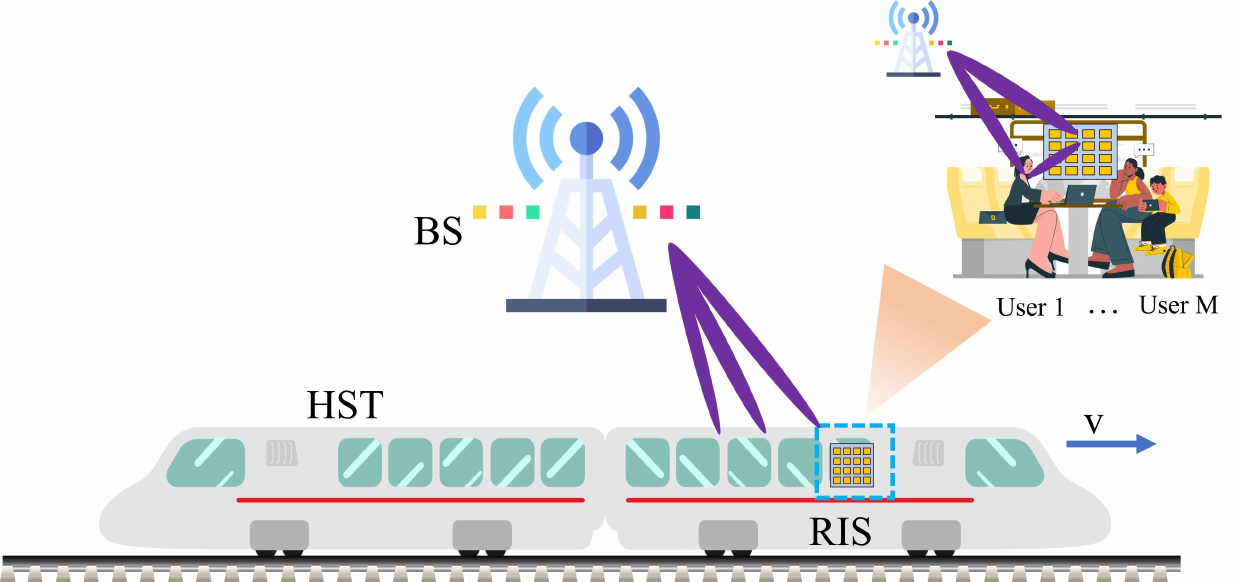}\label{fig:1}}
  \caption{ \label{fig:1} Refracting RIS-assisted URLLC system in mmWave  HST communications.}
\end{figure}

\subsection{Channel Model}
In high-speed mobility environments, a strong line-of-sight (LoS) path and a small amount of reflected and scattered multipath components are usually present. Consequently, the Rician K-factor can be used for modeling these channels, and the Rician fading channel has also been widely adopted in HST communications as in \cite{c5,c6}. Thus, this paper adopts the Rician fading channel. In the proposed system model, the direct link BS-User and the reflection link BS-RIS-User exists at the same time, where the reflection link BS-RIS-User can be cascaded with the BS-RIS and RIS-User links. The detailed descriptions of BS-User, BS-RIS, and RIS-User channels are given below.

\subsubsection{BS-User Channel}
The channel between BS and user $m$ is denoted as $\mathbf{h}_{\mathrm{d},m} \in \mathbb{C} ^{N_B \times 1}$, which can be given as
\begin{equation} \label{eq:hbm}
  \mathbf{h}_{\mathrm{d},m} =\sqrt{PL_{\mathrm{d},m}}\left( \sqrt{\frac{\kappa _f}{\kappa _f+1}}\overline{\mathbf{h}}_{\mathrm{d},m} +\sqrt{\frac{1}{\kappa _f+1}}\widetilde{\mathbf{h}}_{\mathrm{d},m} \right),
\end{equation}
where $\kappa _{f} \geq 0 $ is the Rician K-factor, $PL_{\mathrm{d},m}$ denotes the distance-dependent path loss, which can be expressed as
\begin{equation}
PL_{\mathrm{d},m} =\left( \frac{\lambda}{4\pi D_{\mathrm{d},m}} \right) ^{-\mathrm{\alpha}_{\mathrm{d},m}},
\end{equation}
where  $D_{\mathrm{d},m}$ denotes the distance between BS and user $m$, $\lambda = \frac{c}{f_c}$ denotes the wavelength, $f_c$ is the carrier frequency, and $c$ represents the light speed, ${\alpha}_{\mathrm{d},m}$ denotes the exponent of path loss, $\overline{\mathbf{h}}_{\mathrm{d},m}\in \mathbb{C} ^{N_B \times 1} $ is the LoS component, given by
\begin{equation}
\overline{\mathbf{h}}_{\mathrm{d},m} =e^{j2\pi f_{d1}\tau}\mathbf{a}_{\mathrm{B}}\left( \phi _{{n,m}}^{\mathrm{BU}} ,\delta _{n,m}^{\mathrm{BU}} \right),
\end{equation}
where $f_{d1}=\upsilon \cos \phi _{{n,m}}^{\mathrm{BU}} \cos \delta _{n,m}^{\mathrm{BU}} /\lambda$ and $\tau$  represent the Doppler frequency shift and the duration of a time slot, respectively, $\phi _{{n,m}}^{\mathrm{BU}}$ and $\delta _{n,m}^{\mathrm{BU}}$ are the azimuth and elevation arrival angles, respectively, $\mathbf{a}_{\mathrm{B}}\left( \cdot ,\cdot  \right) \in \mathbb{C} ^{N_B \times 1}$ denotes the antenna array response vectors associated with the BS. The resultant phase of any array element $\mathbf{a}_{\mathrm{B}}$ is equal to the superposition of the x and y phase difference components, which can be expressed as follows \cite{c8}:
\begin{equation} \label{eq:BS}
\mathbf{a}_{\mathrm{B}}\left( \phi _{n,m}^{\mathrm{BU}},\delta _{n,m}^{\mathrm{BU}} \right) =\mathbf{a}_{\mathrm{B}}^{x}\left( \phi _{n,m}^{\mathrm{BU}},\delta _{n,m}^{\mathrm{BU}} \right) \otimes \mathbf{a}_{\mathrm{B}}^{y}\left( \phi _{n,m}^{\mathrm{BU}},\delta _{n,m}^{\mathrm{BU}} \right).
\end{equation}

Consequently, where $\mathbf{a}_{\mathrm{B}}\left( \phi _{{n,m}}^{\mathrm{BU}},\delta _{n,m}^{\mathrm{BU}}\right)$ is the antenna array response vectors associated with the BS, it can be  written as
\begin{equation}
\begin{array}{l}
  \mathbf{a}_{\mathrm{B}}^{x}\left( \phi _{n,m}^{\mathrm{BU}},\delta _{n,m}^{\mathrm{BU}} \right)  \\
 =\frac{1}{\sqrt{N_{x_{\mathrm{B}}}}}\left[ 1,\cdots , e^{\frac{j2\pi d_{x_{\mathrm{B}}}n_{x_{\mathrm{B}}}}{\lambda}\sin \phi _{n,m}^{\mathrm{BU}}\cos \delta _{n,m}^{\mathrm{BU}}},  \right.   \\   \left.
\cdots ,  e^{\frac{j2\pi d_{x_{\mathrm{B}}}\left( N_{x_B}-1 \right)}{\lambda}\sin \phi _{n,m}^{\mathrm{BU}}\cos \delta _{n,m}^{\mathrm{BU}}} \right] ^T ,
\end{array}
\end{equation}

\begin{equation}
  \begin{array}{l}
    \mathbf{a}_{\mathrm{B}}^{y}\left( \phi _{n,m}^{\mathrm{BU}},\delta _{n,m}^{\mathrm{BU}} \right) \\
    =\frac{1}{\sqrt{N_{y_{\mathrm{B}}}}}\left[ 1,\cdots ,e^{\frac{j2\pi d_{y_{\mathrm{B}}}n_{y_{\mathrm{B}}}}{\lambda}\sin \phi _{n,m}^{\mathrm{BU}}\sin \delta _{n,m}^{\mathrm{BU}}}, \right.   \\   \left.
      \cdots ,e^{\frac{j2\pi d_{y_{\mathrm{B}}}\left( N_{y_{\mathrm{B}}}-1 \right)}{\lambda}\sin \phi _{n,m}^{\mathrm{BU}}\sin \delta _{n,m}^{\mathrm{BU}}} \right] ^T,
  \end{array}
\end{equation}
where $d_{x_{\mathrm{B}}}$ and $d_{y_{\mathrm{B}}}$ are the distance between the uniform antenna array elements along the X-axis and Y-axis respectively, $ N_{x_{\mathrm{B}}}$ and $ N_{y_{\mathrm{B}}}$ represent the number of rows and columns of the UPA in the 2D plane, respectively.

Moreover, $\widetilde{\mathbf{h}}_{\mathrm{d},m}\in \mathbb{C} ^{N_B \times 1}$ represents the multipath components, which is modeled as an independent circularly symmetric complex Gaussian (CSCG) random vector with $\mathcal{C} \mathcal{N} (0,\mathbf{I}_{N_B})$  to characterize small-scale fading.

\subsubsection{Refracting RIS-Assisted Channel}
The channel between the BS and the refracting RIS, and between the refracting RIS and the user $m$ are denoted by $\mathbf{h}_{\mathrm{BR}}\in \mathbb{C} ^{F\times N_B}$ and $\mathbf{h}_{\mathrm{R},m}\in \mathbb{C} ^{F\times 1}$, respectively, given by
\begin{equation}
\mathbf{h}_{_{\mathrm{BR}}}=\sqrt{PL_{\mathrm{BR}}}\left( \sqrt{\frac{\kappa _f}{\kappa _f+1}}\overline{\mathbf{h}}_{\mathrm{BR}}+\sqrt{\frac{1}{\kappa _f+1}}\widetilde{\mathbf{h}}_{\mathrm{BR}} \right),
\end{equation}
\begin{equation}
\mathbf{h}_{\mathrm{R},m}=\sqrt{PL_{\mathrm{R},m}}\left( \sqrt{\frac{\kappa _f}{\kappa _f+1}}\overline{\mathbf{h}}_{\mathrm{R},m}+\sqrt{\frac{1}{\kappa _f+1}}\widetilde{\mathbf{h}}_{\mathrm{R},m} \right),
\end{equation}
where $PL_{\mathrm{BR}}$ and $PL_{\mathrm{R},m}$denote the distance-dependent path losses,  given by
\begin{equation}
  PL_{\mathrm{BR}} =\left( \frac{\lambda}{4\pi D_{\mathrm{BR}}} \right) ^{-\alpha _\mathrm{BR}},
\end{equation}
\begin{equation}
  PL_{_{\mathrm{R},m}}=\,\,\left( \frac{\lambda}{4\pi D_{\mathrm{R},m}} \right) ^{-\alpha _{\mathrm{R},m}},
\end{equation}
respectively, where $\alpha _\mathrm{BR}$ and $\alpha _{\mathrm{R},m}$ are the exponent of path losses, $D_{\mathrm{BR}}$ and $D_{\mathrm{R},m}$ denote the distance between the BS and the refracting RIS, and between the refracting RIS and user $m$, respectively. $\overline{\mathbf{h}}_{\mathrm{BR}} \in \mathbb{C}^{F\times N_B}$ and $\overline{\mathbf{h}}_{\mathrm{R},m}\in \mathbb{C}^{F\times 1}$ denote the LoS components, and $\widetilde{\mathbf{h}}_{\mathrm{BR}}\in \mathbb{C}^{F\times N_B}$ and $\widetilde{\mathbf{h}}_{\mathrm{R},m}\in \mathbb{C}^{F\times 1}$ denote the multipath components. Similarly, the entries of $\widetilde{\mathbf{h}}_{\mathrm{BR}}$ and $\widetilde{\mathbf{h}}_{\mathrm{R},m}$ are independently generated from CSCG distribution  $\mathcal{C} \mathcal{N} (0,\mathbf{I}_{F\times N_B})$ and  $\mathcal{C} \mathcal{N} (0,\mathbf{I}_{F})$, respectively. $\overline{\mathbf{h}}_{\mathrm{BR}}$ is given by
\begin{equation}
  \overline{\mathbf{h}}_{\mathrm{BR}}=e^{j2\pi f_{d2}\tau}\mathbf{a}_{\mathrm{R}}\left( \phi _{f,n}^{\mathrm{R}},\delta _{f,n}^{\mathrm{R}} \right) \mathbf{a}_{\mathrm{B}}^{H}\left( \phi _{f,n}^{\mathrm{B}},\delta _{f,n}^{\mathrm{B}} \right),
\end{equation}
where $\mathbf{a}_{\mathrm{R}}\left( \cdot ,\cdot  \right) \in \mathbb{C}^{F\times 1}$ denotes the antenna array response vectors associated with the refracting RIS, which can be expressed in a similar form as \eqref{eq:BS},  $\phi _{f,n}^{\mathrm{R}} \left( \phi _{f,n}^{\mathrm{B}}\right)$ and $\delta _{f,n}^{\mathrm{R}} \left( \delta _{f,n}^{\mathrm{B}} \right)$ are the azimuth and elevation arrival angles,   respectively, and $f_{d2}=\upsilon \cos \phi _{f,n}^{\mathrm{R}}\cos \delta _{f,n}^{\mathrm{R}}/\lambda$ represents the Doppler frequency shift. Since the refracting RIS is relatively stationary with respect to the user, the influence of Doppler is not considered, $\overline{\mathbf{h}}_{\mathrm{R},m}$ can be given as
\begin{equation}
  \overline{\mathbf{h}}_{\mathrm{R},m}=\mathbf{a}_{\mathrm{R}}\left( \phi _{f,m}^{\mathrm{RU}},\delta _{f,m}^{\mathrm{RU}} \right),
\end{equation}
where $\phi _{f,m}^{\mathrm{RU}}$ and $\delta _{f,m}^{\mathrm{RU}}$ denote the azimuth and elevation arrival angles, respectively.

\subsection{Transmission Model}
According to the aforementioned models,  the equivalent channel between BS and user $m$ can be given as
\begin{equation} \label{eq:h}
  \mathbf{h}_m =\mathbf{h}_{\mathrm{d},m}^{H}+\mathbf{h}_{\mathrm{R},m}^{H}\mathbf{\Theta h}_{\mathrm{BR}},
\end{equation}
and the received signal at user $m$ can be expressed as
\begin{align}
  y_m=\mathbf{h}_m\mathbf{w}_ms_m+\sum_{j=1,j\ne m}^M{\mathbf{h}_m\mathbf{w}_js_j}+n_m,\forall m\in \mathcal{M},
\end{align}

where $\mathbf{w}_m\in \mathbb{C} ^{N_B \times1}$ denotes the transmit beamforming vectors for user $m$, and $s_m$ denotes the signal intended to user $m$, $\forall m\in \mathcal{M}$. Without loss of generality, this study assumes $\mathbb{E}\left\{s_m\right\}=0$ and $\mathbb{E}\left\{|s_m|^2\right\}=1$. $n_m\sim \mathcal{C} \mathcal{N} (0,\sigma ^2)$ is the additive white-Gaussian noise (AWGN). $\mathbf{\Theta }=\mathrm{diag}\left\{ e^{j\theta _1},\cdots ,e^{j\theta _F} \right\} $ denotes a diagonal matrix. $f\in \mathcal{F} \triangleq \left\{ 1,\cdots ,F \right\}$, $\theta _f$ represents the phase shift of the $f$-th refracting RIS element. For simplicity,  this study takes finite discrete values with equal quantization  between $\left[0, 2\pi \right)$. The set of phase shifts at each refracting RIS element is given by $\mathcal{S} =\left\{ 0, \Delta\theta ,\cdots ,\Delta\theta  \left( E-1 \right) \right\} $, where $ \Delta\theta  =\frac{2\pi}{E}$ and $E=2^b$, where $b$ is the number of refracting RIS quantization bits. 

For user $m$, the signal to interference plus noise ratio (SINR) can be represented as 
\begin{equation}
\gamma _m =\frac{\left| \left( \mathbf{h}_{_{\mathrm{d},m}}^{H}+\mathbf{h}_{\mathrm{R},m}^{H}\mathbf{\Theta h}_{\mathrm{BR}} \right) \mathbf{w}_m \right|^2}{\sum_{j\ne m}{\left| \left( \mathbf{h}_{_{\mathrm{d},j}}^{H}+\mathbf{h}_{\mathrm{R},j}^{H}\mathbf{\Theta }\mathbf{h}_{\mathrm{BR}} \right) \mathbf{w}_j \right|^2}+\sigma^{2}},\forall m\in \mathcal{M},
\end{equation}
where $\left| \cdot \right|$ denotes the modulus operation.

The Shannon capacity theorem for the maximum achievable rate is  invalid in FBL transmissions. The maximum achievable rate for FBL codes under quasi-static AWGN channel conditions is given by \cite{c9}
\begin{equation} \label{eq:rr}
R_m=\log _2\left( 1+\gamma _m \right) -\left( \frac{V_m}{L} \right) ^{1/2}Q^{-1}\left( \varepsilon _m \right),
\end{equation}
where $\varepsilon _m$ and $L$ are the packet error probability (PEP) and blocklength of user $m$, respectively, $V_m\triangleq \left( \log _2e \right) ^2\left( 1-\left( 1+\gamma _m \right) ^{-2} \right)$ is the channel dispersion, and $Q^{-1}\left( \cdot\right)$ is the inverse of Q-function $Q\left( x \right) =\frac{1}{\sqrt{2\pi}}\int_x^{\infty}{e^{-\frac{t^2}{2}}dt}$.

\subsection{Problem Formulation}
Denote by $P_{\max}$ the maximum transmit power. Let $\boldsymbol{\theta}=[\theta_1,\cdots,\theta_F]$ and $\mathbf{W}=[\mathbf{w}_1,\cdots,\mathbf{w}_M] \in \mathbb{C}^{N_B\times M}$. This study aims to maximize the system sum rate at users by jointly optimizing the active beamforming matrix $\mathbf{W}$ at the BS, PEP for each user, and discrete phase shifts $\boldsymbol{\theta}$ at the refracting RIS in this paper. The optimization problem cam be formulated as
\begin{align}  \label{eq:P11}
  \mathcal{P}_1: \underset{\mathbf{W},\boldsymbol{\theta },\varepsilon _m}{\max }&\sum_{m=1}^M{R_m}  \\
  \mathrm{s}.\mathrm{t}. \quad & \mathrm{Tr}\left( \mathbf{W}^H\mathbf{W} \right) \le P_{\max} \tag{\ref{eq:P11}a},   \label{yyaa1} \\
  &0\leqslant \varepsilon _m\left( t \right) \leqslant \varepsilon _{m,\max}, \forall m\in \mathcal{M},  \tag{\ref{eq:P11}b} \label{yyaa2} \\
  & \theta _f \in \mathcal{S} =\left\{ 0,\Delta \theta ,\cdots ,\Delta \theta \left( E-1 \right) \right\} ,\forall f\in \mathcal{F},   \tag{\ref{eq:P11}c} \label{yyaa3}
\end{align}
where constraint \eqref{yyaa1} guarantees that the total transmit power does not exceed the maximum transmit power $P_{\max}$ at BS, and constraint  \eqref{yyaa2} restricts the worst-case PEP at user $m$ to $\varepsilon _{m,\max}$, and \eqref{yyaa3} means that the phase shift of each refracting RIS element takes its value from a discrete set $\mathcal{S}$.

\section{Sum Rate Maximization}
The formulated sum rate maximization optimization problem in $\mathcal{P}_1$ is intractable in polynomial time due to the non-concave objective function and the non-convex discrete phase shift constraints for refracting RIS elements. Thus, in this section, this study proposes  a power allocation algorithm based on alternative optimization method. First, this study decouples the joint optimization problem $\mathcal{P}_1$ into two subproblems, i.e., active beamforming design and PEP optimization and discrete phase shift design. Then, this study employs employ Lagrangian dual method and the local search method to solve the sumproblems in an alternating manner.

\subsection{Active Beamforming Design and PEP Optimization Subproblem}
For any given phase shifts $\boldsymbol{\theta}$, the problem of joint design of active beamforming and PEP optimization can be rewritten as
\begin{align}  \label{eq:P222}
  \mathcal{P}_2: \underset{\mathbf{W},\varepsilon _m}{\max }&\sum_{m=1}^M{R_m}  \\
  \mathrm{s}.\mathrm{t}. \quad & \mathrm{Tr}\left( \mathbf{W}^H\mathbf{W} \right) \le P_{\max} \tag{\ref{eq:P222}a},   \\
  &0\leqslant \varepsilon _m\left( t \right) \leqslant \varepsilon _{m,\max}, \forall m\in \mathcal{M}.  \tag{\ref{eq:P222}b}
\end{align}

However, the optimization problem  $\mathcal{P}_2$  still has a non-concave objective function.
The maximum achievable rate in \eqref{eq:rr}  for user $m$ in high SINR regime can be approximately expressed as \cite{c10} 
\begin{equation} \label{eq:P19}
\widetilde{R}_m\approx \log _2\left( \gamma _m \right) -\left( \frac{1}{L} \right) ^{1/2}Q^{-1}\left( \varepsilon _m \right) \log _2e.
\end{equation}

For given $\mathbf{W}$, we introduce the following theorem to obtain optimal $\varepsilon _{m}$ \cite{r44}: 
\begin{theorem}
The optimal PEP $\varepsilon _{m}^*$ must meet $\varepsilon _{m}^* = \varepsilon _{m,\max}, \forall m \in m\in \mathcal{M}$.
\end{theorem}

\begin{proof}
This theorem can be proved by contradiction as follows. The optimal PEP is denoted by $\varepsilon _{m}^*,\forall m$ and meets the strict inequalities, i.e., $\varepsilon _{m}^* < \varepsilon _{m,\max}$. When the constraints in optimization problem \eqref{eq:P222} are satisfied, increasing $\varepsilon _{m}^*$ value can increase the value of the rate. Therefore, this contradicts the optimality of $\varepsilon _{m}^*$.
\end{proof}

Thus, let $\varepsilon _{m}^* = \varepsilon _{m,\max}$ in \eqref{eq:P222}. However, the achievable rate expression in \eqref{eq:P222} is still non-concave. Eliminating MU interference term by using zero-forcing (ZF) beamforming enables further relaxation into a concave form.

With combined channel $\mathbf{h}_{m}$, $m \in \mathcal{M}$, the corresponding ZF constraints are given by $\mathbf{h}_m\mathbf{w}_j = \left( \mathbf{h}_{\mathrm{d},m}^{H}+\mathbf{h}_{\mathrm{R},m}^{H}\mathbf{\Theta} \mathbf{h}_{\mathrm{BR}} \right)\mathbf{w}_j = 0, \forall j \ne m, j \in \mathcal{M} $. Let $
\mathbf{H}=\mathbf{H}_{\mathrm{d}}^{H}+\mathbf{H}_{\mathrm{R}}^{H}\mathbf{\Theta }\mathbf{h}_{\mathrm{BR}}$, where $\mathbf{H}_{\mathrm{d}}^{H}=\left[ \mathbf{h}_{\mathrm{d},1},\dots \mathbf{h}_{\mathrm{d},M} \right] ^{^H}$ and $\mathbf{H}_{\mathrm{R}}^{H}=\left[ \mathbf{h}_{\mathrm{R},1},\dots \mathbf{h}_{\mathrm{R},M} \right] ^{^H}$. With those addition constraints, the ZF active beamformer can be given as
\begin{equation} \label{eq:ww}
\mathbf{W}=\mathbf{H}^H\left( \mathbf{HH}^H \right) ^{-1}\sqrt{\boldsymbol{P}},
\end{equation}
where $\boldsymbol{P}=\mathrm{diag}\left\{ p_1,\dots p_M \right\} $ is the power allocation matrix and $p_m$ denotes the power allocation of user $m$. Especially, the ZF beamforming fills $\left| \mathbf{h}_m\mathbf{w}_m \right|=\sqrt{p_m}$ and $\sum_{m\prime\ne m}{\left| \mathbf{h}_m\left( t \right) \mathbf{w}_{m\prime}\left( t \right) \right|^2}=0$, $\forall m \in\mathcal{M}$. Therefore, the optimized $R_m$ of user $m$ can be expressed as
\begin{equation} \label{eq:R3}
  \overline{R}_m =\log _2\left( p_m/\sigma ^2 \right) -\left( \frac{V_m\left( \gamma _m \right)}{L} \right) ^{1/2}Q^{-1}\left( \varepsilon _m \right).
\end{equation}

The active beamforming design subproblem can be converted into a power allocation subproblem by using \eqref{eq:P19} and \eqref{eq:R3}, which can be given as
\begin{align} \label{eq:P33}
  \mathcal{P}_3: \underset{p_{m\ge 0}}{\max }&\sum_{m=1}^M{\widetilde{\bar{R}}_m}\triangleq \sum_{m=1}^M{\left\{ \log 2\left( \frac{p_m}{\sigma ^2} \right) -\frac{Q^{-1}\left( \varepsilon _m \right) \log _2e}{\sqrt{L}} \right\}} \\ \nonumber
&\mathrm{s}.\mathrm{t}. \quad \mathrm{Tr}\left( \boldsymbol{P}^{1/2}\widetilde{\mathbf{W}}^H\widetilde{\mathbf{W}}\boldsymbol{P}^{1/2} \right) \le P_{\max}, \tag{\ref{eq:P33}a}
\end{align}
where $\widetilde{\mathbf{W}}^H=\mathbf{H}^H\left( \mathbf{HH}^H \right) ^{-1}$. The optimization problem $\mathcal{P}_3$ is convex with respect to $p_m$ and can be addressed via Lagrangian dual method.

Next, this study introduces the Lagrangian dual method in detail to solve the power allocation problem. The Lagrangian function of \eqref{eq:P33} are given as
\begin{align} \label{eq:LL}
  L\left( p_m,\mu \right) &=\sum_{m=1}^M{\tilde{\bar{R}}_m}-\mu \left( \mathrm{Tr}\left( \boldsymbol{P}^{1/2}\widetilde{\mathbf{W}}^H\widetilde{\mathbf{W}}\boldsymbol{P}^{1/2} \right) -P_{\max} \right)  \\ \nonumber
  &=\sum_{m=1}^M{\tilde{\bar{R}}_m}-\mu \left( \sum_{m=1}^M{p_m}\widetilde{\mathbf{w}}_{m}^{H}\widetilde{\mathbf{w}}_m-P_{\max} \right),
\end{align}
where $\mu \geqslant 0$ is the Lagrangian multiplier.

By differentiating $L\left( p_m,\mu \right)$ in \eqref{eq:LL} with respect to $p_m$, this study can obtain the optimal transmit power as follows:
\begin{equation}
p_m^\ast=\left[ \frac{1}{\mu \left( \widetilde{\mathbf{w}}_{m}^{H}\widetilde{\mathbf{w}}_m \right) \ln 2} \right] ^+,
\end{equation}
where $\left[ x \right] ^+=\max \left\{ x,0 \right\}$.

Accordingly, the Lagrange dual function can be written as
\begin{equation}
\mathcal{J} \left( \mu \right) =\min_{p_m} L\left( p_m,\mu \right),
\end{equation}

The problem in \eqref{eq:P33} can then be converted into a Lagrange dual problem:
\begin{align} \label{eq:LLL}
\underset{\mu}{\max } \quad \mathcal{J} \left( \mu \right) \\ \nonumber
\mathrm{s}.\mathrm{t}. \quad \mu \ge 0 .
\end{align}

Since $\mathcal{J} \left( \mu \right)$ is the minimum of linear functions relative to $\mu$ according to \eqref{eq:P33} and \eqref{eq:LL}, we can prove that \eqref{eq:LLL} is invariably convex. Therefore, the solutions of \eqref{eq:LLL} can be derived by utilizing the subgradient projection technology \cite{c13}. Thus, the projected subgradient method for \eqref{eq:LL} is expressed as
\begin{align}
&\mu \left( t_1 +1 \right) \\  \nonumber
&=\left[ \mu \left( t_1 \right) -\omega \left( t_1 \right) \left( P_{\max} - \mathrm{Tr}\left( \boldsymbol{P}^{1/2}\widetilde{\mathbf{W}}^H\widetilde{\mathbf{W}}\boldsymbol{P}^{1/2} \right) \right) \right] ^+,
\end{align}
where $t_1$ denotes the iteration index. The parameters $\omega $ is small enough for non-negative step size. The proposed Lagrange dual method-based power allocation algorithm is represented in Algorithm~\ref{PA}.

\begin{algorithm}[!t]
  \caption{Lagrange Dual Method-based Power Allocation Algorithm}
  \label{PA}
  \begin{algorithmic}[1]
  \REQUIRE
  Set the maximum number of iterations $T_1$ and convergence threshold $\varepsilon _p$, and initialize Lagrange multipliers $\mu \left( t_1 \right)$ for $t_1=0$
  \ENSURE
    {$p _m^\ast,~\forall m \in \mathcal{M}$ }
  \REPEAT
  \STATE Compute power allocation strategy
  \\
  $p_m=\left[ \frac{1}{\mu \left( \widetilde{\mathbf{w}}_{m}^{H}\widetilde{\mathbf{w}}_m \right) \ln 2} \right] ^+$;
  \STATE Update the Lagrange multipliers:\\
  $ 
  \begin{aligned}
    &\mu \left( t_1 +1 \right) \\  \nonumber
    &=\left[ \mu \left( t_1 \right) -\omega \left( t_1 \right) \left( P_{\max} - \mathrm{Tr}\left( \boldsymbol{P}^{1/2}\widetilde{\mathbf{W}}^H\widetilde{\mathbf{W}}\boldsymbol{P}^{1/2} \right)   \right) \right] ^+,
    \end{aligned};
  $
  \IF{$\left| \mu \left( t_1+1 \right) -\mu \left( t_1 \right) \right| < \varepsilon _p$}
  \STATE  Convergence = \textbf{true};
  \RETURN $p_m^\ast=p_m$;
  \ELSE
  \STATE $t_1=t_1+1$;
  \ENDIF
  \UNTIL Convergence = \textbf{true} or $t_1=T_1$;
  \end{algorithmic}
\end{algorithm}

\subsection{Phase Shift Design Subproblem}
Under the given active beamforming $\mathbf{W}$ and PEP $\varepsilon_{m}$, the problem for phase shift design $\boldsymbol{\theta}$ at refracting RIS can be formulated as
\begin{align}  \label{eq:P333}
  \mathcal{P}_4: \underset{\boldsymbol{\theta}}{\max }&\sum_{m=1}^M{R_m}  \\ \nonumber
  \mathrm{s}.\mathrm{t}. \quad & \theta _f \in \mathcal{S} =\left\{ 0,\Delta \theta ,\cdots ,\Delta \theta \left( F-1 \right) \right\} ,\forall f\in \mathcal{F}
\end{align}

Due to the non-convexity of problem in \eqref{eq:P333}, it is challenging to solve by convex optimization. Considering the complexity of the problem, the simplest straightforward method to solve this problem is exploiting exhaustive search. However, it is obvious that the complexity of exhaustive algorithms increases exponentially as the set of practical solutions increases \cite{c14}. Considering the complexity and effectiveness, this study utilizes the local search method to solve the discrete phase shift optimization problem as shown in Algorithm~\ref{phase} \cite{r13}. Specifically, when optimizing the phase shift of refracting RIS element $f$, keeping the other $F-1$ phase shift values fixed. Then for the phase shift $\theta_f$ of each  refracting RIS element, we traverse all possible values in the discrete phase shift set and select the optimal value $\theta_f^\ast$ that satisfies the optimization objective, i.e., maximizing the sum rate. Then, we use this optimal solution $\theta_f^\ast$ to overwrite the original value of $\theta_f$ when optimizing another phase shift, until all phase shifts in the set $\mathcal{S}$ are fully optimized.
\begin{algorithm}[!t]
  \caption{Local Search-based Algorithm for Phase Shift Design}
  \label{phase}
  \begin{algorithmic}[1]
  \REQUIRE
 number of refracting RIS elements $F$ and quantization bits $b$
  \ENSURE
    {$\theta _f^\ast,~\forall f \in \mathcal{S}$ }
  \FOR {$f=1:F$}
  \FOR {$\theta _f\in \mathcal{S}$}
  \STATE Assign all possible values to $\theta _f$, and choose the value maximize the sum rate in problem \eqref{eq:P333}   denoted as $\theta _f^\ast$;
  \STATE $\theta _f=\theta _f^\ast$;
  \ENDFOR
  \ENDFOR
  \end{algorithmic}
\end{algorithm}

\begin{algorithm}[!t]
  \caption{Sum Rate Maximization}
  \label{sumrate}
  \begin{algorithmic}[1]
  \REQUIRE
  Set the maximum number of iterations $T_1$, and convergence threshold $\varepsilon _p$, $\xi$, initialize Lagrange multipliers $\mu \left( t_1 \right)$ for $t_1=0$, number of RIS elements $F$, number of quantization bits $b$, $\varepsilon _{m}= \varepsilon _{m,\max}$, the initial iteration index $i = 1$
  \ENSURE $R^{\ast}, \boldsymbol{\theta}^{\ast}, \mathbf{W}^{\ast}$
  \STATE For given $\boldsymbol{\theta}_i$, update $\mathbf{W}_{i+1}$ using \eqref{eq:ww} by Algorithm~\ref{PA};
  \STATE For given $\mathbf{W}_i$, update $\boldsymbol{\theta}_{i+1}$  by Algorithm~\ref{phase};
  \IF{$\left|R^{(i+1)}-R^{(i)}\right|<\xi $}
  \STATE $R^{\ast}=R^{(i+1)}$;
  \STATE $\boldsymbol{\theta}^{\ast}=\boldsymbol{\theta}_{i+1}$;
  \STATE $\mathbf{W}^{\ast}=\mathbf{W}_{i+1}$;
  \ELSE
  \STATE $i = i +1$, and return Step 1;
  \ENDIF
  \end{algorithmic}
\end{algorithm}

\subsection{Sum Rate Maximization}
According to the aforementioned, this study summarizes the power allocation subproblem and refracting RIS discrete phase shift optimization subproblem algorithms, and propose a solution that maximize the sum rate. As shown in Algorithm~\ref{sumrate}, this study sets the initial state of the algorithm randomly, and then alternately update the value of active beamforming matrix and refracting RIS discrete phase shift until $\left|R^{(i+1)}-R^{(i)}\right|<\xi $, the algorithm convergences.

Given that we are considering a HST communication scenario, which involves very rapid movement, the computing speed of the central processing unit (CPU) executing the optimization algorithm becomes a significant physical limitation. Utilizing a more advanced CPU will execute the algorithm faster, increasing the likelihood of achieving optimal performance.

\subsection{Complexity Analysis}
First, there is the active beamforming design and PEP optimization subproblem. According to Theorem 1, the optimal value of PEP $\varepsilon _m$ is the threshold $\varepsilon _{m,\max}$ of the constraint condition, so there is no need to optimize iteratively, which is equivalent to a fixed value. We only need to optimize the beamforming $\mathbf{W}$. After mathematical transformation, the optimization problem is converted into a power allocation subproblem. The frequency of Lagrange multiplier updates and the power allocation to all users in each iteration creates complexity. The number of uses is $M$, and the number of iterations to meet the convergence condition $\left| \mu \left( t_1+1 \right) -\mu \left( t_1 \right) \right| < \varepsilon _p$  is $T_{\mathrm{inner}}$. Because the Lagrange multiplier is updated by the subgradient method, if the precision of the subgradient method is $\boldsymbol{\varpi }$, then its complexity is $\mathcal{O} \left( 1/\varpi ^2 \right) $. Thereby, the complexity of Algorithm~\ref{PA} is $\mathcal{O} \left( T_{\mathrm{inner}}M/\varpi ^2 \right)$. For the RIS phase shift design subproblem, the local search algorithm changes the value of element $f$ while keeping the other phase shifts constant and choosing the best one out of $E = 2^b$  possible ones, and updates a value for $\theta_f$. Due to the number of RIS elements is $F$,  the complexity of Algorithm~\ref{phase} is  $\mathcal{O} \left( EF \right)$. The complexity of the method for maximizing the sum rate is not only determined by the number of iterations, but also in terms of the two subproblems. Let $I_{\mathrm{outer}}$ be the number of iterations that satisfy the convergence condition $\left|R^{(i+1)}-R^{(i)}\right|<\xi $, we obtain the complexity of sum rate maximization algorithm is  $\mathcal{O} \left( I_{\mathrm{outer}}\left( T_{\mathrm{inner}}M/\varpi ^2+EF \right) \right)$.

\section{Performance Elevation}
In this section, this study numerically evaluates the proposed refracting RIS-aided MU-MISO downlink URLLC system for mmWave HST. This study provides numerical simulation results to verify the superiority of the proposed refracting RIS-aided sum rate maximization approach for mmWave HST communication systems. Table~\ref{tab:1} summarizes the considered simulation parameters for the systems.

\subsection{Compared Schemes}
For comparison, this study considers the following benchmark schemes:

\begin{itemize}
  \item {\bf Ideal Phase Shift:} the RIS phase shift is aligned with the phase shift of the direct link channel and the cascaded channel, i.e., $\theta _f = \arg(\mathbf{h}_{_{\mathrm{d},m}}^{H})-\arg(\left[h_{\mathrm{R},m}^{H} \right]_f )-\arg(\left[\mathbf{h}_{\mathrm{BR}} \right]_f \mathbf{w}_m )$, where $\left[h_{\mathrm{R},m}^{H} \right]_f$ is the $f$th element of $\left[\mathbf{h}_{\mathrm{R},m}^{H} \right]_f$ and $\left[\mathbf{h}_{\mathrm{BR}} \right]_f$ is the $f$th row vector of  $\mathbf{h}_{\mathrm{BR}}$.
  \item {\bf Shannon Rate:} to obtain an upper bound on the sum rate, the Shannon rate formula is accepted in this paper, i.e., $V_m = 0$. The resulting optimization problem is addressed utilizing a modified version of the proposed scheme.
  \item {\bf Shannon Rate with Ideal Phase Shift:} to obtain an upper bound on the sum rate, the Shannon rate formula is accepted in this paper, i.e., $V_m = 0$. The power allocation problem is solved in the same way as in the proposed scheme, except that the RIS phase shift adopts an ideal phase shift method.
  \item {\bf Binary Search:} the binary search method is used to obtain the appropriate Lagrangian multiplier, and the other parts are the same as the proposed scheme.
  \item {\bf Random Phase Shift:} for each refracting RIS element phase shift $\theta _f$, it can be randomly assigns the value in the range $\left[0, 2\pi \right)$ and keeps it unchanged. We adopt optimized beamforming for the BS.
  \item {\bf Without RIS:} this scheme does not consider the deployment of refracting RIS for signal reflection.
\end{itemize}

\begin{table*}[!t]
  \caption{Simulation Parameters}
  \label{tab:1}
  \centering
  \begin{tabular}{|c|c|c|c|c|c|} \hline
    \textbf{Parameter}  & \textbf{Symbol}    & \textbf{Value} & \textbf{Parameter}  & \textbf{Symbol}    & \textbf{Value}     \\
    \hline
     { Convergence threshold}        & $\varepsilon _p, \xi$   & $10^{-5}$ & Carrier frequency   & $f_c$   & $28$ GHz  \\
    \hline
    \ {Speed of HST}       & $v$      &  $360$ km/h   & Transmit power & $P_{\max}$  &  $30$ dBm \\
    \hline
     {Number of antennas}       & $N_B$     &  $4$   & Noise power  & $\sigma^2$   & \makecell{ $-174~ {\rm{dBm/Hz}}$ \\ $+10\log _{10}B+10 ~\rm{dB}  $} \\
    \hline
     {Number of RIS elements}        & $F$      &  $36$  & {Maximum PEP} & $\varepsilon _{m,\max}$  & $10^{-4}$ \\
    \hline
     Blocklength  & $L$  & $100$  &  Rician $K$-factor  & $\kappa_f$ & $10$ dB \\
    \hline
    {Number of users} & $M$  &   $4$  & Number of phase quantization bits  &  $b$  & $2$  \\
    \hline
    {System bandwidth}    & $B$     & $200$ MHz & Path loss exponent  & $\alpha_{\left( {\rm{(d,m),(R,m),BR}}\right)}$ & $4, 3, 2$  \\
    \hline
  \end{tabular}
  
\end{table*}

\subsection{Convergence of Proposed Algorithm}
First, this study certifies the convergence behavior of the proposed algorithm for addressing the joint active beamforming and refracting RIS discrete phase shift design problem. As shown in Fi{}g.~{\ref{fig:2}}, the proposed algorithm is effective and illustrates the convergence rates of the six schemes, excluding the binary search. The performance of the proposed algorithm significantly surpasses both the random phase shift scheme and the scheme without RIS. The value of the equivalent channel is a superposition of complex numbers. The random phase shift method randomly assigns phase shifts to each element of the RIS, which may result in the reflection link at the RIS being exactly opposite to the direct link. This opposition weakens the gain of the combined channel, leading to poorer performance compared to the proposed algorithm. In the absence of RIS deployment, there is no superposition of RIS reflection links, only direct links, resulting in even worse performance than both the proposed algorithm and the random phase shift scheme. The curve considering the Shannon rate provides an upper bound of the sum rate.  However, the rate expression used in the proposed scheme yields a lower value compared to the Shannon rate due to packet decoding errors, resulting in lower performance than the Shannon rate scheme. Additionally, the performance of the ideal phase scheme is observed to be better than that of the Shannon rate scheme. This is because the value of the Shannon capacity component in the rate expression \eqref{eq:rr} under the ideal phase scheme, referred to as $\log _2\left( 1+\gamma _m \right)$, is greater than the rate value of the Shannon rate scheme under the proposed algorithm. The value of the subtracted part in \eqref{eq:rr} is relatively small, much smaller than the difference in $\log _2\left( 1+\gamma _m \right)$ between the two cases. Therefore, the overall performance of the ideal phase scheme surpasses that of the Shannon rate.  Furthermore, the ideal phase shift case outperforms the RIS discrete phase shift optimized case. When the Shannon rate adopts the ideal phase, the upper bound of the ideal phase scheme is achieved. It can also be observed that, except for the binary search, the total rates of the six schemes converge within approximately 21 iterations, with similar convergence speeds. This further verifies the effectiveness of the proposed joint active beamforming and discrete phase shift optimization scheme. The total rate of the binary search scheme converges within about 17 iterations, and the final converged value matches that of the proposed algorithm. However, its value does not exhibit a stable trend in the initial iterations; it first increases and then decreases before converging. In contrast, the trend of the proposed algorithm steadily.
\begin{figure}[!t]
  \centering
  {\includegraphics[scale=0.6]{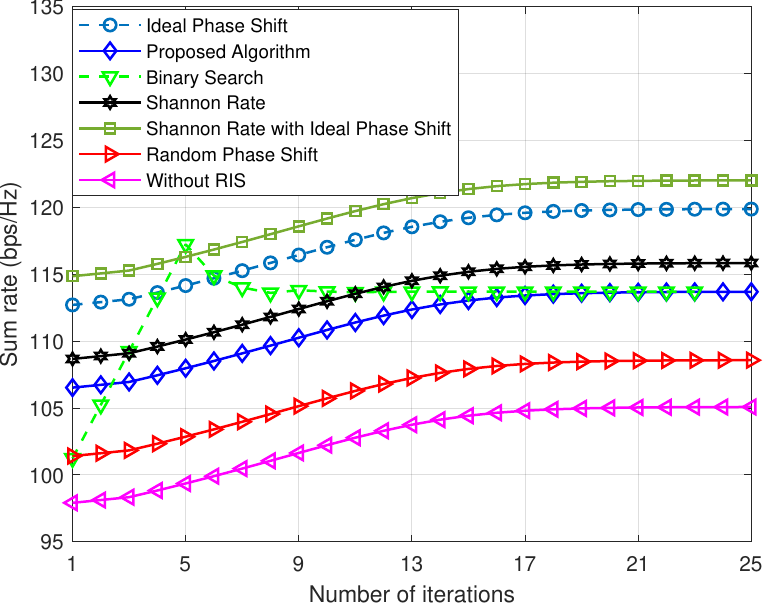}}
  \caption{ \label{fig:2}Sum rate versus number of iterations for different schemes.}
\end{figure}
\begin{figure}[!t]
  \centering
  {\includegraphics[scale=0.6]{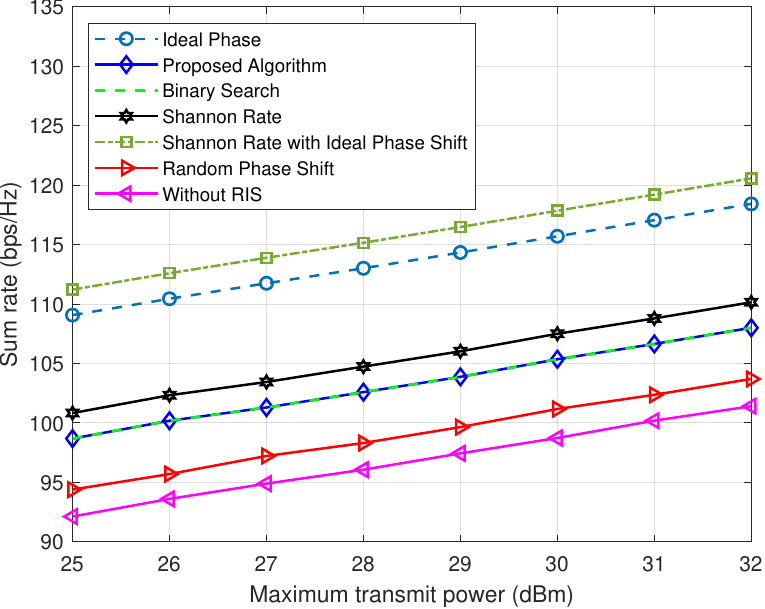}}
  \caption{ \label{fig:3}Sum rate versus transmit power $P_{\max}$ for different schemes.}
\end{figure}

\subsection{Impact of Transmit power}
Fi{}g.~{\ref{fig:3}} illustrates the sum rate against the transmit power $P_{\max}$. It is observed that the higher transmit power from the BS leads to better system performance across all seven cases. Notably, the proposed algorithm consistently outperforms both the random phase shift scheme and without RIS scheme, except when compared to the ideal phase shift, the Shannon rate with ideal phase shift and the Shannon rate scheme. Additionally, the performance of the binary search scheme similarly to the proposed algorithm. This similarity arises because binary search only differs in the method of obtaining the Lagrangian multiplier, yet the final convergence value remains the same as that of the proposed algorithm as shown in Fi{}g.~{\ref{fig:2}}, resulting in equivalent performance. In the current scenario, the Shannon rate curve provides an upper bound on the rate but cannot guarantee the high reliability and low latency demands of the system. This is because the Shannon rate does not consider the performance loss due to SPT in power allocation optimization, resulting in a power allocation policy that may exceed power limitations. The performance gap between the Shannon rate scheme and the proposed scheme remains nearly constant as transmit power increases. Similarly, the performance gaps between all other schemes remains nearly constant. For given $P_{\max} = 30$ dBm, the sum rate of the proposed algorithm is $6.7\%$ higher than that of the system without RIS deployment. Additionally, for a sum rate of approximately $100$ bps/Hz, the proposed algorithm requires $16.13\%$ less transmit power compared to the scheme without refracting RIS. This implies that refracting RIS deployment can not only improve the system coverage capacity but also reduce transceiver power consumption and manufacturing costs.

\begin{figure}[!t]
  \centering
  {\includegraphics[scale=0.6]{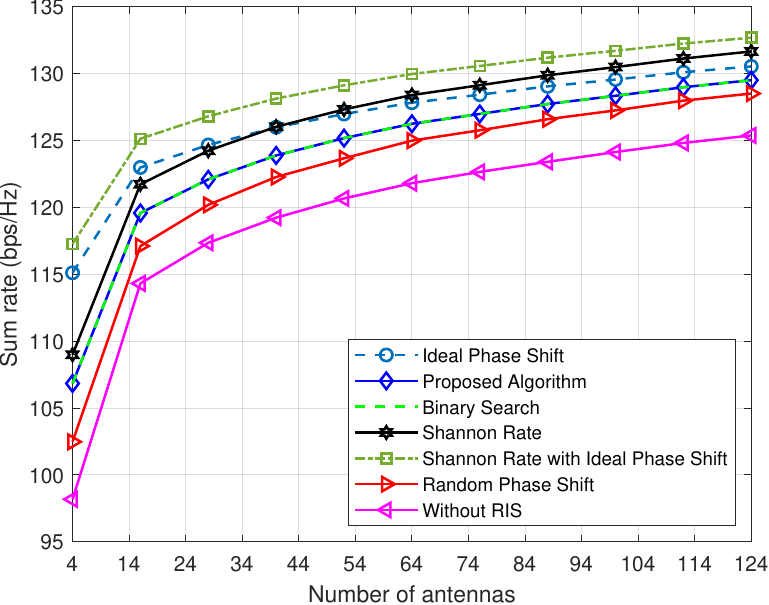}}
  \caption{ \label{fig:4}Sum rate versus number of antennas $N_B$ for different schemes.}
\end{figure}
\subsection{Impact of Number of Antennas}
Fi{}g.~{\ref{fig:4}} shows the sum rate versus the number of antennas $N_B$. As can be seen from Fi{}g.~{\ref{fig:4}}, the sum rate under each of the seven schemes increases with $N_B$, due to the higher diversity gain provided by more antennas. The proposed algorithm is superior to both the random phase shift scheme and without RIS scheme. Under the same conditions, when $N_B=4$, $N_B=34$, $N_B=64$ and $N=94$, the sum rate of the proposed algorithm is  $4.26\%$, $1.32\%$, $0.92\%$ and $0.78\%$ higher than the random phase shift, respectively. This demonstrates that the performance gap between the proposed algorithm and the random phase shift narrows as $N_B$ increases. This occurs because the received signal of the direct link becomes stronger and the impact of the refracting RIS discrete phase shift optimization becomes weaker with an increasing number of antennas. In other words, the performance of the random phase shift RIS can also be assured when the number of antennas is large enough, but this increases the maintenance cost of the antenna arrays. This also underscores the necessity of using RIS with optimized phase shifts. Additionally, the performance of the proposed algorithm approaches that of the ideal phase shift scheme as $N_B$ increases. Similarly, the performance of the Shannon rate scheme also approaches the Shannon rate scheme with ideal phase shift as $N_B$ increases. This indicates that the higher diversity gain brought by an increasing number of antennas can compensate for the performance difference between the discrete RIS phase shift and the ideal phase shift, with this difference diminishing as the number of antennas increases. Therefore, it can be inferred that as the number of antennas further increases, the performance of the proposed algorithm could approximate that of the ideal phase shift scheme. An interesting observation is that when the number of antennas is less than $44$, the performance of the ideal phase shift scheme is better than that of the Shannon rate scheme, though the performance gap gradually narrows. The reason for this trend is similar to the explanation provided in Fi{}g.~{\ref{fig:2}}. However, when the number of antennas exceeds 44, the Shannon rate scheme outperforms the ideal phase shift scheme, with the performance gap gradually widening as the number of antennas increases. This further demonstrates that the performance difference between the discrete RIS phase shift and the ideal phase shift decreases with an increasing number of antennas. It also suggests that when the number of antennas reaches a certain threshold, the performance of the Shannon rate scheme can match that of the ideal phase shift scheme. As the number of antennas continues to increase, the performance of the Shannon rate scheme surpasses that of the ideal phase shift scheme. This occurs because, with more antennas, the received signal from the direct link becomes stronger, thereby weakening the influence of the RIS reflection link.

\begin{figure}[!t]
  \centering
  {\includegraphics[scale=0.6]{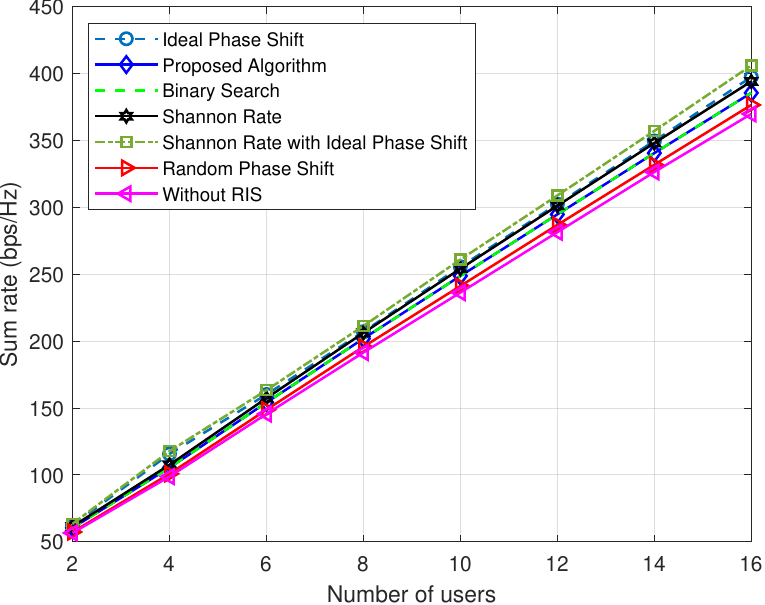}}
  \caption{ \label{fig:5}Sum rate versus number of users $M$ for different schemes.}
\end{figure}
\subsection{Impact of Number of Users}
Fi{}g.~{\ref{fig:5}} illustrates the sum rate versus the number of users $M$. It can be observed that the sum rates increase with $M$ across all seven schemes. The proposed scheme outperforms both the random phase shift and the scheme without RIS. Additionally, the performance gap between the proposed scheme and the Shannon rate increases with $M$. This phenomenon arises because the increasing number of users intensifies channel time dispersion, thereby hindering communication quality and slowing down performance improvements, resulting in a widening gap. A similar effect occurs between our proposed scheme and the random phase shift due to the random selection of phase shifts potentially degrading system performance. The growing gap between the proposed scheme and the scheme without RIS deployment is attributed to the increasing enhancement in channel gain through RIS, which accelerates system performance improvements with $M$, underscoring the significance of RIS deployment.

\begin{figure}[!t]
	\centering
	{\includegraphics[scale=0.6]{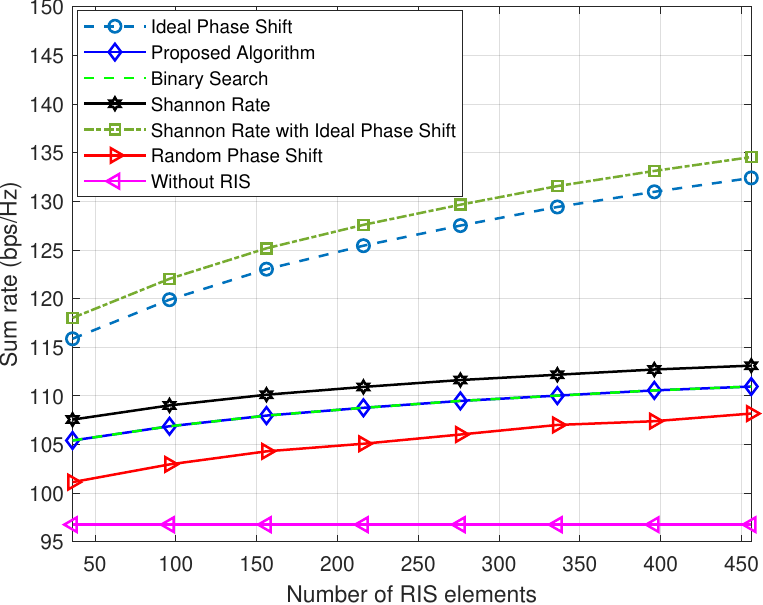}}
	\caption{ \label{fig:6}Sum rate versus number of RIS elements $F$ for different schemes.}
\end{figure}
\subsection{Impact of Number of RIS Elements}
Fi{}g.~{\ref{fig:6}} shows the effect of number of refracting RIS elements on the sum rate of the four schemes, by varying number of refracting RIS elements $F$ from $36$ to $456$. From the results, it can be observed that the three cases with RIS are superior to the case of without RIS deployment, with their sum rates increasing as the number of RIS element $F$ increases. This improvement is due to the fact that adding RIS elements produce sharper beams, which can more effectively refract the incident signals to the target user, enhance the received power at users, and reduce the interferences among users. When $F=216$, $F=276$, $F=336$, $F=396$, and $F=456$, the sum rate of the proposed scheme increases by around $0.64\%$, $0.5\%$, $0.47\%$ and $0.39\%$, respectively. This demonstrates that this upward trend gradually slows down as $F$ increases, indication that under the same conditions, the sum rate does not increase infinitely as the number of  RIS elements increases. The sum rate provided by the RIS reflection link is mainly limited by transmit power, system model, and physical environment. Additionally, we can also observe that the performance gap between the proposed algorithm and the ideal phase shift becomes larger as $F$  increases. This is because the performance gain brought by the ideal phase shift case becomes larger as $F$  increases.

\subsection{Impact of Number of RIS Quantization Bits}
Fi{}g.~{\ref{fig:7}} illustrates the influence of the number of quantization bits $b$ on the sum rates. From the results, we can see that the sum rates of the proposed algorithm, the binary search scheme, and the Shannon rate scheme increase with $b$. This is due to the fact that phase resolution increases as the number of quantization bits increases, which enhances the channel gain at users and results in an improvement in the sum rate. As expected, the performance of the ideal phase scheme, the Shannon rate scheme with ideal phase, and the scheme without RIS is not affected by $b$. a significant increase in the sum rate for the Shannon rate scheme, the binary search scheme, and the proposed algorithm as $b$ increases from $1$ to $2$. However, when $b$ exceeds $2$, the growth rate of the sum rate for the proposed algorithm, the binary search scheme, and the Shannon rate scheme gradually slows down as $b$ increases. Additionally, the performance gap between the proposed algorithm and the ideal phase shift scheme gradually decreases as $b$ increases. This implies that performance of the proposed algorithm can better approach the ideal phase shift case when $b \rightarrow \infty $.
\begin{figure}[!t]
  \centering
  {\includegraphics[scale=0.6]{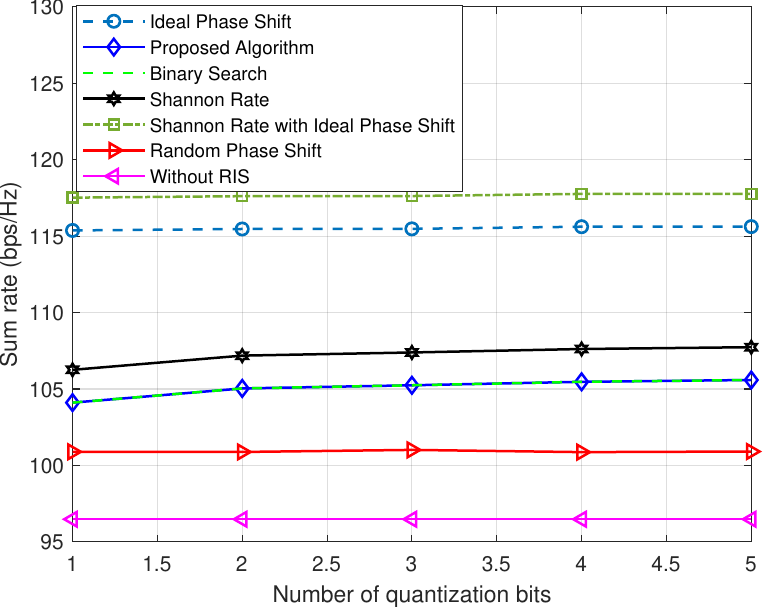}}
  \caption{ \label{fig:7}Sum rate versus number of RIS quantization bits $b$ for different schemes.}
\end{figure}

\subsection{Impact of Blocklength}
Fi{}g.~{\ref{fig:8}} shows the impact of  blocklength $L$ on the sum rate of the seven schemes, by varying blocklength $L$ from $100$ to $600$. As expected, the curve of the Shannon rate with ideal phase shift scheme and the Shannon rate scheme remains flat, and the sum rate performance of the other five cases increases with blocklength. This is because increasing the blocklength is equivalent to increasing the number of channel uses, thus more data packets can be transmitted per unit time, i.e., the transmission rate increases. From a mathematical perspective, the second term of \eqref{eq:rr} becomes smaller as $L$ increases, so the numerical result of the rate becomes larger, thereby increasing the system sum rate. It can also be observed that the sum rates of the five schemes, except for the Shannon rate with ideal phase shift scheme and the Shannon rate scheme, significantly improve when $L$ increases from $100$ to $200$. When $L$ is greater than $200$, the upward trend of the sum rate performance of the other five cases gradually slows down. This implies that as $L \rightarrow \infty $, the sum rates of the other four cases with RIS discrete phase shift will eventually converge to the Shannon rate scheme, and the sum rate of the ideal phase shift scheme will eventually converge to the Shannon rate with ideal phase shift scheme. Thus, the URLLC system's use of SPT can effectively reduce the latency of data transmission, but it sacrifices part of the rate performance. Additionally, if packet loss occurs, long packet transmission experience higher communication overhead compared to SPT, SPT generally outperforms long packet transmission. It can also be seen from the figure that the blocklength has a significant effect on system performance in short packet communication systems, and the communication system design based on SPT is essentially different from that based on high data rate long packet transmission.
\begin{figure}[!t]
  \centering
  {\includegraphics[scale=0.6]{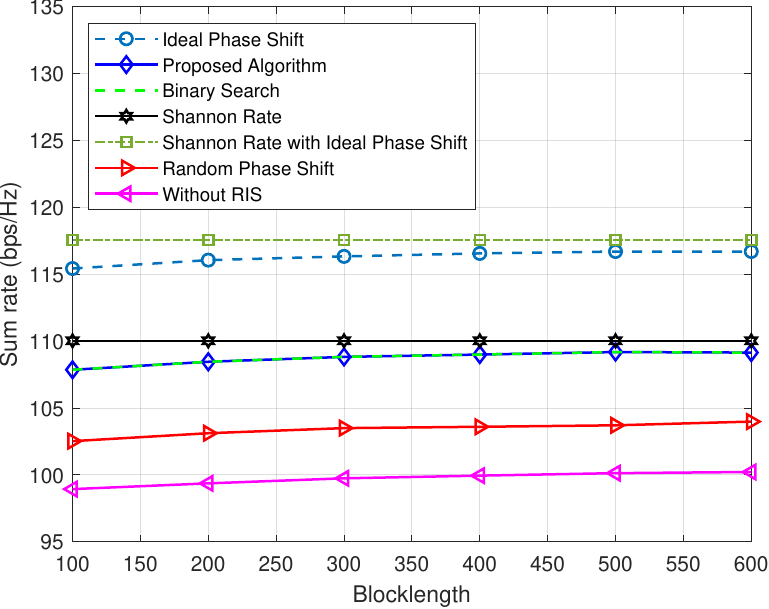}}
  \caption{ \label{fig:8}Sum rate versus blocklength $L$ for different schemes.}
\end{figure}

\begin{figure}[!t]
	\centering
	{\includegraphics[scale=0.6]{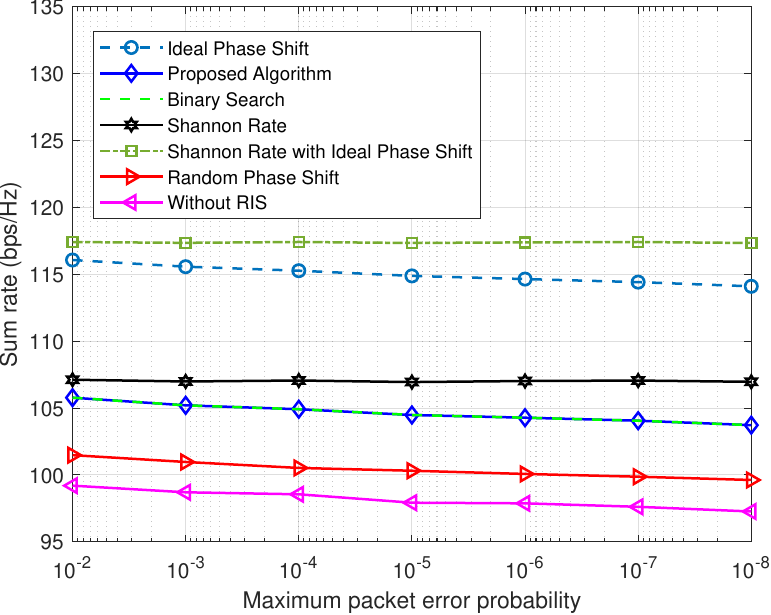}}
	\caption{ \label{fig:9}Sum rate versus maximum PEP $\varepsilon _{m,\max}$ for different schemes.}
\end{figure}
\subsection{Impact of Maximum PEP}
Fi{}g.~{\ref{fig:9}} illustrates the sum rate versus the maximum PEP $\varepsilon _{m,\max}$. As expected, the curve of the Shannon rate with ideal phase shift scheme and the Shannon rate scheme remains flat, while the sum rate of the other fives schemes decreases $\varepsilon _{m,\max}$ decreases from $10^2$ to $10^8$. This is because increasing $\varepsilon _{m,\max}$ means that the system's communication reliability  becomes lower, reducing the transmission requirements for data packets, and allowing more data packets to be successfully decoded, thereby increasing the transmission rate and improving the system rate. In other words, the higher reliability requirements of the communication system, the more pronounced the system performance degradation. Especially when $\varepsilon _{m,\max}$ is reduced from $10^{-2}$ to $10^{-8}$, the decrease in sum rate is more obvious. This indicates  that the improvement of system reliability will inevitably lead to the loss of part of the system performance. When $\varepsilon _{m,\max} = 10^{-2}$, the sum rate of the proposed algorithm is improved by about $4.2\%$ compared to the random phase shift and about $6.47\%$ compared to the case without RIS. Therefore, for a given $\varepsilon _{m,\max}$, the proposed algorithm has significant advantages in improving the sum rate. Moreover, under the same sum rate, the proposed algorithm has a lower PEP value, meaning that deploying refractive RIS can enhance the reliability of the system.

\subsection{Impact of HST Speed}
Fi{}g.~{\ref{fig:10}} illustrates the influence of HST speed $v$ on the sum rate of the seven schemes. From the results, we can see that the sum rates under the seven schemes increase with $v$. This is because HST can cover longer distance with a faster speed under the same time duration, bringing it closer to the BS, thereby increasing the channel gain, which enhances the received power at users and improves the rate. We can also see that system sum rate of the proposed scheme is superior to the scheme without RIS by $6.61\%$ at $v =100$ km/h and by $7.14\%$ at $v =1000$ km/h. Therefore, for a given $v$, the proposed algorithm can improve the sum rate. Furthermore, the proposed algorithm can achieve the same system performance at a lower speed. This demonstrates the necessity of refracting RIS deployment.
\begin{figure}[!t]
	\centering
	{\includegraphics[scale=0.6]{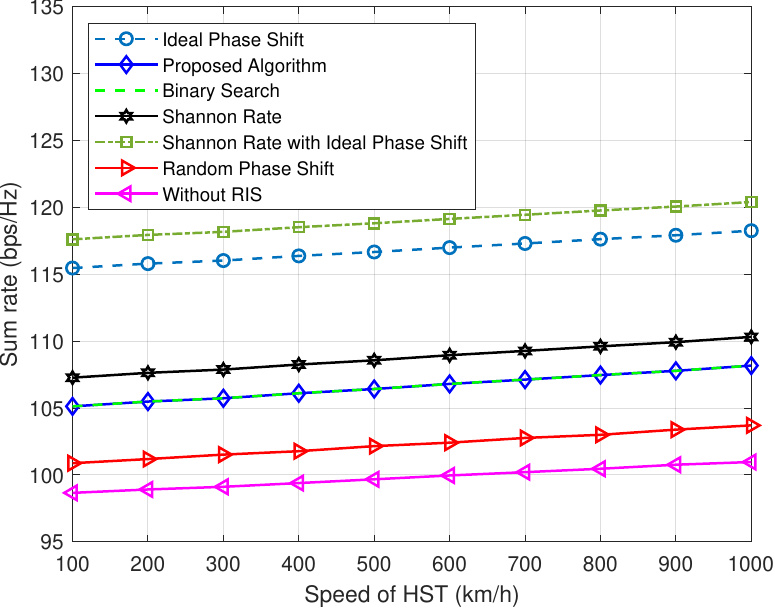}}
	\caption{ \label{fig:10}Sum rate versus HST speed $v$ for different schemes.}
\end{figure}
\begin{figure}[!t]
	\centering
	{\includegraphics[scale=0.6]{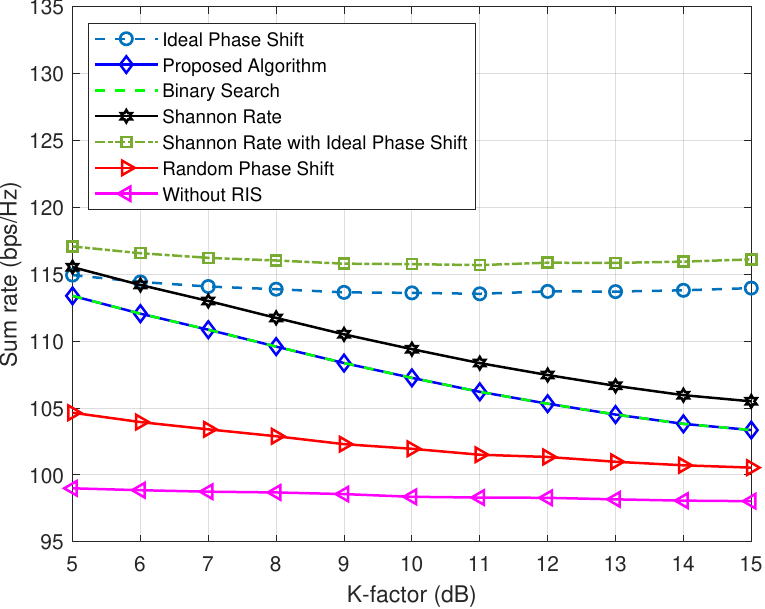}}
	\caption{ \label{fig:11}Sum rate versus K-factor $\kappa_f$  for different schemes.}
\end{figure}
\subsection{Impact of K-factor}
Fi{}g.~{\ref{fig:11}} illustrates the sum rate versus the Rician K-factors. As can be seen Fi{}g.~{\ref{fig:11}}, the sum rate under each of the seven schemes decreases with $\kappa_f$. It can also be observed that as $\kappa_f$ increases, the performance gap between the five schemes becomes narrower except for the ideal phase shift scheme and the Shannon rate scheme with ideal phase shift, where binary search exhibits the same performance as the proposed algorithm. This is because a larger $\kappa_f$ increases result in a stronger LoS component and better channel quality, while the reduction of adjustable refraction multipath limits the system capacity improvement brought by refracting RIS deployment. Additionally, it is evident that the performance gap between the ideal phase shift scheme and the proposed algorithm widens as the increase of $\kappa_f$. A similar trend is observed between the Shannon rate scheme with ideal phase shift and the Shannon rate scheme with discrete phase shift. This is because a larger $\kappa_f$ enhances the LoS component, while the impact of RIS gradually diminishes. However, the degree to which the RIS impact on performance is reduced in the ideal phase shift scheme and the Shannon rate scheme with ideal phase shift is significantly lower than in the proposed algorithm and the Shannon rate scheme, and this gap gradually increases as $\kappa_f$  continues to rise.

\section{Conclusions}
In this paper, we focused on a MU-MISO downlink URLLC in a refracting RIS-assisted HST communication system coverage enhancement. We proposed a sum rate maximization problem of joint active beamforming and RIS discrete phase shift optimization, while meeting the power, discrete phase shift, and reliability constraints. We designed a joint optimization algorithm based on the iterative optimization method. Furthermore, we investigated the influences of several crucial system parameters including transmit power, number of RIS elements, number of users, number of antennas, number of RIS quantization bits, blocklength, maximum PEP, speed of HST, and Rician K-factor on sum rate of HST communication systems. Simulation results show that deploying RIS can effectively improve system performance and verify the effectiveness of the proposed algorithm in terms of sum rate. In future work, we will consider more advanced and low-complexity methods for beamforming design and RIS phase shift optimization to increase the theoretical depth of this study \cite{r9new,r9new2}.

\end{document}